\documentclass[12pt]{article}
\pdfoutput=1
\usepackage[utf8]{inputenc}
\usepackage[margin=1in]{geometry} 
\usepackage{biblatex}
\usepackage[ruled, vlined, linesnumbered]{algorithm2e}
\usepackage{amsmath}
\usepackage{amssymb}
\usepackage{amsthm}
\usepackage{hyperref}
\usepackage[nameinlink, capitalise]{cleveref}
\newtheorem{theorem}{Theorem}[section]
\newtheorem{corollary}{Corollary}[theorem]

\newtheorem{lemma}[theorem]{Lemma}
\theoremstyle{definition}

\theoremstyle{remark}

\newcommand{\T}{\mathcal{T}}
\newcommand{\R}{\mathcal{R}}
\newcommand{\nos}[2]{\text{nos}_{#1}^{#2}}
\newcommand{\pos}[2]{\text{pos}_{#1}^{#2}}

\hypersetup{
    colorlinks=true,
    linkcolor=blue,  
    urlcolor=blue,
}

\title{A $D$-competitive algorithm for the Multilevel Aggregation Problem with Deadlines}
\author{
Jeremy McMahan \thanks{Dept. of Computer Science, University of Wisconsin-Madison. \href{mailto:jmcmahan@wisc.edu}{\texttt{jmcmahan@wisc.edu}}}}
\begin{document}

\maketitle

\begin{abstract}
    In this paper, we consider the multi-level aggregation problem with deadlines \\ (MLAPD) previously studied by Bienkowski et al. \cite{OG}, Buchbinder et al. \cite{D}, and Azar and Touitou \cite{GF}. This is an online problem where the algorithm services requests arriving over time and can save costs by aggregating similar requests. Costs are structured in the form of a rooted tree. This problem has applications to many important areas such as multicasting, sensor networks, and supply-chain management. In particular, the TCP-acknowledgment problem, joint-replenishment problem, and assembly problem are all special cases of the delay version of the problem.

    We present a $D$-competitive algorithm for MLAPD. This beats the $6(D+1)$-competitive algorithm given in Buchbinder et al. \cite{D}. Our approach illuminates key structural aspects of the problem and provides an algorithm that is simpler to implement than previous approaches. We also give improved competitive ratios for special cases of the problem. 
    %The analysis pairs the critically overdue framework from \cite{OG} with a notion of phases to understand the evolution of the algorithm compared to the optimal. Unlike previous works that use a top-down analysis approach, our analysis uses a hybrid approach with both top-down and bottom-up aspects in order to avoid any waste in our charging argument. 
\end{abstract}
%\setcounter{page}{0}
%\thispagestyle{empty}
%\newpage

\section{Introduction}

In many optimization contexts, aggregating tasks and performing them together can be more efficient. For example, consider sending multiple gifts to a relative. If there are multiple presents, it may be cheaper to bundle them all in one large box rather than sending each present individually. Although, aggregation is useful in many contexts, it is often constrained. For the gift example, the postal service may have weight restrictions on packages so bundling items together may be impossible. In such a situation, multiple smaller bundles may have to be constructed resulting in a complex optimization problem. The type of aggregation constraints we will be interested in form a tree-like structure that commonly appear in communication networks.

%For instance, it is  more time efficient to stock up on supplies from a store rather than purchasing an item whenever the need arises. In fact, purchasing items in bulk is usually cheaper as well. As another example, consider

%The type of aggregation constraints we will be interested in can be easily seen in the context of communication networks. In every business, communication amongst employees is critical to ensure operations are running smoothly. Though, employees in such companies are usually organized into hierarchies that restricts communication. For instance, a cashier likely does not have the contact information of the CEO and vice versa. Rather, the cashier can only communicate with other employees at a local branch. In principle, the cashier could access the CEO indirectly by sending a message to the local branch's manager and then the manager send a message to another higher-up and so on until the CEO is reached. In this case, the aggregation constraints correspond to the communication network being a tree in which nodes can only communicate directly with their parent. These kinds of aggregation constraints will be the focus of this paper.

% Alternatively, there may be some delay cost that accrues the longer the task is put off.

Besides the aggregation constraints, another difficulty in these problems are the urgency of the tasks. A task might have a hard deadline of when it must be completed or may accrue a cost the longer it is delayed. At the same time, delaying a task allows more tasks to pile up which better exploits the power of aggregation. This illustrates a natural trade-off between the aggregation benefits of delaying tasks and the cost of doing so. Generally, we will refer to strategies that mostly delay tasks to maximize aggregation benefits as \textit{lazy-aggregating} and to strategies that prioritize the urgency of the tasks as \text{eager-aggregating}. 

To illustrate these ideas, consider the example of a communication tree of employees at a restaurant. A health inspector may continually issue fines to a host at a restaurant until some concern is resolved. Depending on the scope of the concern, it may indicate issues more broadly in the company and so would need to be delivered to a higher official than just the local manager. However, the cost of the fines might not be as expensive as the cost for the host to send a message up the chain of command. So, at any time, the host will have to decide whether to send the fine now or wait to receive more messages to send together. In fact, there are many problems of a similar flavor involving transmitting control packages in communication networks \cite{Gathercast, SchemesCMP, TCPOPTRand, wiresensornet}. Similarly, many optimization problems in supply chain management face similar considerations \cite{LotSizeMSA, lotsizingbook, ProductionPlanning, DynamicLotSize}.
%As a more positive example, the message to the host could be an offer for two restaurants to collaborate in catering an event. In this case, it would be natural for the other restaurant to set a deadline for when the offer must be accepted. Similarly, we could model a delay cost that involves the likelihood of missing this opportunity as time passes. 

%In the deadline case, it is tempting to think there is no benefit to sending a message before it is due. However, this is not the case. Again, suppose an employee just received an offer from a company that wants to conduct a mutually beneficial project with this store. The employee's manager may be about to send a message to the higher ups. If the employee sends the message now, it could be very cheap since all the higher levels will have aggregate overlap with the other message. However, the employee could also receive a huge number of additional offers before the deadline and so the 'breadth' savings overcomes the 'depth' savings. Without knowing the future, figuring out the optimal solution is impossible, but approximations are tractable. 

The \textit{online multi-level aggregation problem with delay} (MLAP) introduced in \cite{OG} captures many of the scenarios mentioned above. In the deadline version (MLAPD) \cite{D}, which is a special case of the delay version, requests arrive over time and must be served between their time of arrival and the time they are due. The aggregation constraints are modeled using a node-weighted rooted tree and a service involves transmitting a subtree that contains the root. All requests in a service subtree are satisfied to model aggregation. Then, the goal is to minimize the total cost of all services. MLAP encapsulates many classic problems such as the \textit{TCP-acknowledgment Problem}, the \textit{Joint Replenishment Problem} (JRP), and the \textit{Assembly Problem}.

%This definition captures many useful problems. A request being moved up the tree can represent tasks being performed where prerequisite tasks are handled first and other tasks may be completed together. For example, MLAPD captures the assembly problem which involves building some device where parts must be put together in a particular order. In fact, many classic optimization problems are captured by MLAP for a particular choice of $D$, the number of levels in the tree, which we also refer to as the depth of the tree. For $D = 1$, MLAP is equivalent to the \textit{TCP-acknowledgment problem}. For $D = 2$, MLAP is equivalent to the \textit{joint replenishment problem} (JRP) which is a important problem in supply chain management. 

%Expand on TCP-ak and JRP and assembly.

\subsection{Our Contributions}

In this paper, we present a $D$-competitive algorithm for MLAPD. This beats the previous best algorithm constructed in \cite{D}, which was roughly $6(D+1)$-competitive. Also, the proposed algorithm attacks the problem directly rather than relying on a reduction to $L$-decreasing trees. This illuminates key structural aspects of the problem and results in the algorithm being simpler to implement than previous approaches.

The main analysis used is based off the framework of critically overdue nodes introduced in \cite{OG}, but this framework is paired with a generalization of phases introduced in \cite{D2} to achieve a more refined analysis. In addition, using the ideas of late and early nodes within the phases context provides additional structure that can be exploited. Overall, the analysis uses a charging strategy that uses the concept of node investments to charge nodes proportionally to their impact on a service. The strategy exploits both top-down and bottom-up views of a service in order to maximize the precision of the charges. Investments also allow charging future nodes without causing complications. This in contrast to past approaches that only considered nodes available at the current moment or earlier. 

Additionally, we give improved competitive algorithms for special cases of MLAPD. In particular, we present an algorithm with competitive ratio strictly less than $4$ for paths. This shows that the known lower-bound of $4$ for infinite graphs does not hold for finite graphs. The main analysis used exploits a notion of phases that utilizes the special structure of path instances. This approach involves looking at nested subpaths and then charging several such subpaths to a single path of the optimal schedule.

\subsection{Related Work}

MLAP and MLAPD were originally formulated by \cite{OG}. The offline versions of MLAP and MLAPD have been well studied. Even for $D = 2$, MLAPD and so MLAP is known to be APX-hard \cite{OGJRP, JRPHard, ProductionPlanning}. When the tree is a path, MLAP can be solved efficiently \cite{5P}. For MLAPD, the best approximation factor is $2$ \cite{SensorNets, OG}. Using an algorithm for the assembly problem from \cite{PDI}, \cite{Priv} constructs a $(2 + \epsilon)$-approximation for MLAP where $\epsilon > 0$ is arbitrary. 

For the online variant, the best competitive ratio for MLAP is $O(D^2)$ due to an algorithm constructed in \cite{GF}. For MLAPD, the best competitive ratio known was $6(D+1)$ due to an algorithm developed in \cite{D}. The best known lower bound for MLAPD and MLAP for finite depth trees is $2$, which comes from the lower bound of $2$ on JRP with deadlines (JRPD) given in \cite{D2}. If the depth of the tree is unbounded, the best lower bound is $4$ as shown in \cite{OG} by constructing a hard instance where the tree is a path of arbitrary length. This improves on the lower bound of $2 + \phi$ presented in \cite{5P} for trees of unbounded depth. Additionally, \cite{OG} shows that $4$ is in fact the optimal competitive ratio for paths of unbounded length improving on the results in \cite{5P, 8P}.

There is even more known about MLAP for small $D$. For $D = 1$, MLAP is equivalent to the TCP-acknowledgment problem which involves sending control messages in a single packet in a network. The offline version of the problem is equivalent to the lot sizing problem \cite{DynamicLotSize} and can be solved efficiently \cite{ImprovedLotSize}. For the online variant, it is known that the optimal deterministic competitive ratio is $2$ \cite{TCPDelay}. The optimal competitive ratio when randomness can be used is known to be $\frac{e}{e-1}$ \cite{TCPOPTRand, Seiden00aguessing, OnlPDApproach, AdRev}. For the deadline case, a simple greedy algorithm is exactly $1$-competitive. 

For $D = 2$, MLAP is equivalent to  JRP. JRP tries to optimize the total shipping cost from suppliers to retailers when a warehouse is used as an intermediary place for storage. Although JRPD and so JRP is APX-hard \cite{OGJRP, JRPHard, ProductionPlanning}, JRP permits an $1.791$-approximation \cite{D2} which improves on the results in \cite{OneWareMultiRetail}. In addition, JRPD permits a $1.574$-approximation \cite{OGJRP}. In the online variant, \cite{OPD} gives a $3$-competitive algorithm for JRP which improves upon the $5$-competitive algorithm given in \cite{8P}. This competitive ratio is nearly tight as the best known lower bound for the competitive ratio of JRP is $2.754$ \cite{D2}. For JRPD, the optimal competitive ratio is exactly $2$ \cite{D2}.

\section{Preliminaries}

\paragraph{Multi-level aggregation problem.} The \textit{multi-level aggregation problem with deadlines} (MLAPD) is defined by a pair, $(\T, \R)$, where $\T$ is a node-weighted tree rooted at a node $r$ and $\R$ is a set of requests that arrive at the nodes of $\T$ over some time period. A request $\rho = (v, a, d)$ is specified by an arrival time $a$, a deadline $d$, and a node at which the request is issued, $v \in \T$. Without loss of generality, we can assume that the deadlines of each request are distinct and that the cost of every node is positive. We define $D$ to the height of the tree plus one or, equivalently, the maximum number of nodes present in a path in $\T$. We call $D$ the depth or number of levels of the tree. Then, we define the level of a node, $L_v$, to be the number of nodes on the unique $r$ to $v$ path in $\T$. 

A service is a pair $(S, t)$ where $S$ is a sub-tree of $\T$ containing the root and $t$ is the time that $S$ is transmitted. Notice the definition of a service implies that if $v \in S$ and $u$ is an ancestor of $v$, then $u \in S$. We refer to this fact as the \textit{subtree property} of services. We sometimes refer to the subtree $S$ as the service when the time of the service is clear or irrelevant to the given context. A request is satisfied by a service, $(S,t)$, if $v \in S$ and $a \leq t \leq d$. A schedule is a set of services $(S_1,t_1),\ldots, (S_k,t_k)$. A schedule is a solution to MLAPD if every request in $\R$ is satisfied by some service in the schedule. The cost of a schedule is $\sum_{i = 1}^k c(S_i)$ where $c(S_i) = \sum_{u \in S_i} c(u)$ is the total cost of nodes in the service. Note that this node-weighted definition is equivalent to the original definition of the problem as shown by \cite{D}.

In the online setting, $\T$ is known to the algorithm but the request arrive in an online fashion. In particular, if $\rho = (v,a,d) \in \R$ is a request, then the algorithm receives this request at time $a$. At time $t$, the algorithm must decide whether to transmit a service and if so decide what nodes to include in the service given only knowledge of requests whose arrival time are before or at time $t$. The algorithms we consider only keep track of active requests, i.e. requests that have arrived yet have not been served previously. 

For a node $v$, we let $\T_v$ denote the subtree of $\T$ rooted at $v$. Given a request $\rho$ and a time $t$, we use the notation $\rho \in \T_v^t$ to mean that $\rho$ was issued at a node in $\T_v$ by time $t$ and has not yet been satisfied. Then, if $\rho \in \T_v^t$, we define $P(v \to \rho)$ to be the nodes on the path in $\T_v$ from $v$ to the node at which $\rho$ is issued. For the root, we use the special notation $P_{\rho}$ to be $P(r \to \rho)$. If $S$ is any subtree, then we define $P(v \to \rho | S) = P(v \to \rho) \setminus S$ to be the nodes on the path to $\rho$ in $\T_v$ that are not already in $S$. We also define $c(v \to \rho | S) = c(P(v \to \rho | S))$ to be the cost of the nodes of the path to $\rho$ that are not in $S$.

%\paragraph{$L$-decreasing trees.} An $L$-decreasing tree satisfies the property that along any root to node path, the cost of a node is at least a $L$-factor smaller than the cost of the ancestor of the node on the path. In particular, if $u$ is an ancestor of $v$ in $\T$, we have that $c(u) \geq L c(v)$. This definition is merely a relaxed version of HSTs. It is shown in \cite{OG} that an $R$-competitive algorithm for \textsc{MLAPD} on $L$-decreasing trees implies the existence of a $DLR$-competitive algorithm for \textsc{MLAPD} on general trees using a reduction. For this reason, may algorithmic approaches to \textsc{MLAPD} focus on the case of $L$-decreasing trees. We similarly define $L$-increasing trees to be trees satisfying the property that $c(v) \geq L c(u)$ whenever $u$ is an ancestor of $v$ in $\T$.

\paragraph{Critically overdue.} It is known that there is always an optimal schedule that only transmits at times that are deadlines of requests. Given such an optimal schedule and a schedule of some algorithm, we can define a notion of a node being late. Suppose $(S,t)$ is a service of the algorithm in the given schedule and that $v \in S$. Then, we define $d^t(v) = \min_{\rho \in \T_v^t} d_{\rho}$ to be the deadline of the earliest request in $\T_v^t$, which we call the deadline of $v$ or the urgency of $v$. Also, we define $\nos{v}{t}$ to be the time of the next service of the optimal schedule that includes $v$ and is transmitted at a time strictly after time $t$. Note, we sometimes also use $\nos{v}{t}$ to reference this service of the optimal schedule instead of the time at which it was transmitted. If no other services of the optimal schedule include $v$ at time $t$ or later, then we define $\nos{v}{t} = \infty$. Now, we say that a node $v$ is late at time $t$ if $d^t(v) < \nos{v}{t}$. This is so called since the optimal schedule must have already satisfied the earliest due request in $v$ as the next time this schedule includes $v$ this request will have already been due. Next, we say $v$ is \textit{critically overdue} at time $t$ if $v$ is late at time $t$ and if there is no service of the algorithm's schedule including $v$ that is transmitted at a time in $(t,\nos{v}{t})$ \cite{OG}. We often will say the pair $(v,t)$ is critically overdue to mean $v$ is critically overdue at time $t$.

\section{Special Tree Classes}

In this section we present new upper bounds on the competitive ratio for MLAPD when restricted to special classes of graphs. The arguments in this section serve as a warm up for the main argument in the next section. Also, the new upper bound for path instances shows that the known lower-bound of $4$ on the competitive ratio does not hold for finite graphs.

\subsection{Increasing Trees}

An increasing tree is a tree where weights continually increase as one travels away from the root down any path in the tree. We consider an algorithm that serves only the path to a due request. In particular, whenever a request becomes due, the algorithm transmits exactly the path to the due request and nothing else. Hence, every service is of the form $(P_{\rho}, d_{\rho})$. We call this algorithm \textsc{Noadd} since it never adds any nodes to a service other than the bare minimum.

\begin{theorem}
\textsc{Noadd} is $D$-competitive for any Increasing Tree of depth $D$. 
\end{theorem}

\begin{proof}

We use a simple charging argument. Suppose $\rho$ is a request that causes a service of ALG at time $t = d_{\rho}$. Further, suppose that $v$, the node at which $\rho$ was issued, is at level $L$ and $r = v_1, \ldots,v_L = v$ are the nodes on the path $P_{\rho}$. Then, by definition, we know that at the time of its satisfaction by ALG, $\rho$ is due and so must have been served by OPT by this time. Suppose $(O,t')$ was the service of OPT that satisfied $\rho$. We charge the entire cost of this service of ALG just to the occurrence of $v$ appearing in $O$. We note that $\sum_{i = 1}^L c(v_i) \leq \sum_{i = 1}^Lc(v) = Lc(v)$ since the tree is increasing. Then, the charge to $v$ is exactly $\frac{Lc(v)}{c(v)} = L$. Since $L \leq D$, the total charge to $v$ is at most $D$. Also, since all the requests at $v$ are cleared after this transmission, the next time a request issued at $v$ becomes due must be strictly after the current time $t$. By the same argument, at that time $t'$ a new service of OPT must have been transmitted that contains $v$ and it will be this occurrence of $v$ that is charged next. Thus, all the charges are disjoint and so the charge to any node of OPT is at most $D$. Thus, ALG is $D$-competitive.

\end{proof}

Note that the argument above is actually proving $v$ is critically overdue at time $d_{\rho}$. Then, the charging scheme boils down to charging critically overdue nodes. We can use the same ideas to get an even better result for $L$-increasing trees. $L$-increasing trees simply satisfy the property that if $u$ is an ancestor of $v$ then $c(v) \geq L c(u)$.

\begin{corollary}
\textsc{Noadd} is $\frac{L}{L-1}$-competitive for any $L$-increasing Tree of depth $D$. 
\end{corollary}

\begin{proof}

We use the same argument as for increasing trees but refine the charge to any node. To avoid confusion we use $\ell$ to mean the level and $L$ to be the factor of increase. With the same setup as before, the only change is the total charge to $v$. Note that $c(v) \geq L c(v_{\ell-1}) \geq L^2 c(v_{\ell-2}) \geq \ldots \geq L^{\ell-1} c(v_1)$ since the tree is $L$-increasing. In particular, $c(v_i) \leq L^{i - \ell} c(v)$. Then,
$$\sum_{i = 1}^{\ell} c(v_i) \leq \sum_{i = 1}^{\ell} L^{i-\ell} c(v) = c(v)\sum_{i = 0}^{\ell-1} L^{-i} = \frac{L - L^{1-\ell}}{L-1}c(v) \leq \frac{L}{L-1}c(v)$$ 
So, the total charge to $v$ is $\frac{L}{L-1}$. The rest of the argument goes through as before. Thus, ALG is $\frac{L}{L-1}$-competitive.

\end{proof}

%Overall, we see that solving \textsc{MLAPD} on increasing trees is straightforward and becomes easier the faster the costs of the nodes grow. Similarly, the problem for $L$-decreasing trees seems to get easier the faster the node costs shrink as it becomes more feasible to add many nodes to a service. Despite this similarity, $L$-decreasing trees do not seem to yield small constant competitive ratios. In fact, solving the problem on $L$-decreasing trees may be as hard as solving the problem in general.

\subsection{Paths}\label{section: paths}

In our context, we consider a path to simply be a path graph where one of the endpoints is the root of the tree. Consider the algorithm, \textsc{Double}, that upon a request $\rho$ becoming due, first adds the subpath from the root to the node at which $\rho$ was issued, $P_{\rho}$. Then, the algorithm continually tries to extend the current subpath containing the root by adding other subpaths that connect the last node in the current subpath to nodes holding urgent requests. If adding a subpath would cause the resultant service to cost more than $2c(P_{\rho})$, then the path is not added and the algorithm transmits the current service. This iterative procedure is described by \cref{alg: double}. 
%Alternatively, suppose $\gamma_1$ is the earliest due request satisfying $c(r \to \gamma_1 | P_{\rho}) \leq c(P_{\rho})$ but $c(r \to \gamma_2 | P_{\rho}) > c(P_{\rho})$ where $\gamma_2$ is the next earliest due request after $\gamma_1$. Then, the service constructed will be exactly $(P_{\rho} \cup P(r \to \gamma_1 | P_{\rho}), d_{\rho})$.

\begin{algorithm}[H]
\caption{\textsc{Double}}\label{alg: double}
\SetAlgoLined
\KwData{Request $\rho$ just became due}
\KwResult{A subtree is transmitted that serves $\rho$ }
 Initiate a service $S = P_{\rho}$.\\
 \While{there is an unsatisfied request}{
  Let $\gamma$ be the earliest due request not yet satisfied by $S$\\
  \If{$c(S) + c(r \to \gamma | S) > 2c(P_{\rho})$}{
  break;
  }
  $S \gets S \cup P(r \to \gamma | S)$\\
 }
\textbf{transmit} $S$
\end{algorithm}

\begin{theorem}\label{theorem: paths}
\textsc{Double} is $(4 - 2^{-D})$-competitive for any path with depth $D$.
\end{theorem}

\begin{proof}
We will use a charging argument that charges the cost of several services of ALG to a single service of OPT. We visualize each subpath as an interval going from the root to another node. Then, this strategy boils down to charging sequences of nested intervals of ALG to a single interval of OPT. 

Let $(S,d) = (P_{\rho} \cup A, d_{\rho})$ be a service of ALG. Since $\rho$ is due at time $d_{\rho}$, OPT must have already satisfied $\rho$ at some earlier time $t$. If $(O,t)$ is the service of OPT satisfying $\rho$, then we know that $P_{\rho} \subseteq O$ by the subtree property of services. Also, the subtree property and the fact that the tree is a path implies that either $O \subseteq S$ or $S \subseteq O$. Since each cost is positive, this corresponds to $c(O) \leq c(S)$ in the first case and $c(S) \leq c(O)$ in the second case. 

\begin{itemize}
    \item Suppose that $c(S) \geq c(O)$ so that ALG's interval contains OPT's interval. By definition of ALG, we have that $c(A) \leq c(P_{\rho})$. Thus, we charge $(S,d_{\rho})$ to $(O,t)$ at a total charge of
    $$c(S) = c(P_{\rho}) + c(A) \leq 2c(P_{\rho}) \leq 2c(O) $$
    Here, we used the fact that $P_{\rho} \subseteq O$ and so $c(P_{\rho}) \leq c(O)$. 
    We claim that no other service $(S',d') = (P_{\gamma} \cup B, d_{\gamma})$ of ALG satisfying $c(S') \geq c(O)$ will be charged to $(O,t)$. For the sake of contradiction, suppose that $(S',d')$ and $(S,d)$ are both charged to $(O,t)$. WLOG, we assume that $d_{\gamma} > d_{\rho}$. Note, it must be the case that $\gamma$ and $\rho$ both arrive before time $t$ for $O$ to satisfy them. Thus, when $S$ was constructed, $\gamma$ had already arrived. Since $\gamma$ triggered a service, we know it was not satisfied by $S$. Hence, the stopping conditions of ALG imply that $S_i \subset P_{\gamma}$. But, $S$ containing $O$ then implies that $O \subseteq S \subset P_{\gamma}$ so $O$ could not have satisfied $\gamma$, a contradiction. Thus, only one service of ALG with larger cost may be charged to any service of OPT.

    \item Suppose that $c(S) < c(O)$ so that ALG's interval is strictly contained in OPT's interval. We will pay for all such $S$ at once. Suppose that $\rho_1,\ldots, \rho_k$ are requests in order of increasing deadline that are all satisfied by $(O,t)$ and that each trigger a service of ALG having smaller service cost than $c(O)$. Let $(S_i, t_i) = (P_{\rho_i} \cup A_i, d_{\rho_i})$ be the services of ALG triggered by $\rho_i$. By \cref{lemma: doubling} we know that 
    $$\sum_{i = 1}^k c(S_i) <  (2 - 2^{-k})c(O)$$
    and so the charge to $O$ for all these intervals is $2 - 2^{-k}$. 
\end{itemize}

Overall, we see any service of OPT is charged at most $2$ larger services of ALG and at most $2-2^{-k}$ for all the smaller services of ALG. Thus, the total charge made to any service of OPT is at most $2 + (2-2^{-k}) = 4 - 2^{-k}$. Hence, the algorithm is $4 - 2^{-k}$-competitive, where $k$ is the maximal number of smaller services that we charge to any service of OPT. A trivial upper bound for $k$ is $D$ as each interval must differ by at least one node. Thus, ALG is $(4 - 2^{-D})$-competitive.

\end{proof}

\begin{lemma}\label{lemma: doubling}
Suppose that $(O,t)$ is a service of OPT and $\rho_1,\ldots, \rho_k$ are requests in order of increasing deadline that are satisfied by $(O,t)$ and that trigger services $(S_1,t_1),\ldots, (S_k,t_k)$ of ALG each having smaller service cost than $c(O)$. Then,
$$\sum_{i = 1}^k c(S_i) <  (2 - 2^{-k})c(O)$$
\end{lemma}
We defer the proof of \cref{lemma: doubling} to section~\ref{sec: pathappendix} of the appendix. In fact, the analysis above can be refined to give a $(4 - 2^{1 - \frac{D}{2}})$-competitive ratio. The proof of this stronger result can also be found in section~\ref{sec: pathappendix} of the appendix.

\begin{corollary}
The competitive ratio lower bound of $4$ does not hold for finite paths. Thus, the true competitive ratio for MLAPD for path instances is either at most $3$ or non-integer. 
%In fact, the lower bound of $4$ does not hold for infinite graphs unless the request sequence is infinite and contains a subsequence of requests such that the corresponding subsequence of distances from the issued nodes to the root are increasing and unbounded.
\end{corollary}

\begin{proof}
The claims are immediate as $4 - 2^{-D} < 4$ for any finite $D$. 
%So long as $k$ is finite the upper bound given is strictly less than $4$. Only infinitely many distinct nodes having requests could yield non-finite $k$ by the definition of $k$. Thus, a sequence of nodes with increasing distance from the root are needed to cause $k$ to go to infinity. Alternatively, this can be seen since if all nodes having requests have a bounded distance from the root, we can truncate the path to the node of longest distance. Hence, we get an equivalent instance on a finite path by truncation, and then can apply \cref{theorem: paths} directly.
\end{proof}

\section{Arbitrary Trees}

In this section, we present a $D$-competitive algorithm for MLAPD on arbitrary trees of depth $D$. 

%To devise such an algorithm, there are a few important considerations. The main one is whether we prefer eager or lazy aggregation. We attempt to have a balance of both. In particular, we allow nodes already included in a service to include others, but we give them a set budget so that not too many nodes are added. When adding early requests, we must decide whether a local or global approach is preferable. Global approaches use a node's budget to satisfy the most urgent requests regardless if they are located near the node in question. These approaches maximally deal with the time requirement parts of the problem. On the other hand, local approaches look to satisfy requests that share a large fraction of the path from the root to the given node in order to maximize the power of aggregation. In a sense, local strategies are desirable since if we already have paid for a large cost path, we would regret having to pay for that path again in the future rather than taking advantage of that path now to satisfy nearby requests. Though, it turns out aspects of both local and global strategies are needed.

\paragraph{Intuition.} Since the path to a due request must always be served, the main task is to decide what other nodes will be added to the service. The main observation is that several of these nodes have been included in ALG unnecessarily in the past. So, we need to justify their inclusion in the current service. For a node $v$, we do this by satisfying requests in $\T_v$ using a procedure called a $v$-fall. This way if $v$ was indeed serviced too many times before, we have other nodes that can effectively be used to pay for $v$'s cost in the analysis. Repeating this logic, we need to recursively justify all the nodes added by the $v$-fall, which will be nodes of higher level than $v$. Eventually, we find good nodes of low level, or this process reaches the leaves, which will always have the desired properties. 

Clearly, for this to work the fall will have to include enough new nodes to pay for the cost of the given node. At the same time, adding too many nodes would runs the risk of too much eager-aggregation which would just cause more of the same problem we seek to avoid. To deal with this, we used a fixed budget that determines how many nodes can be added by the fall. However, the path to a request in $\T_v$ may have cost far greater than $v$'s budget. To avoid $v$'s budget going to waste, we discount the cost of this path so that it becomes easier to include by other falls. We treat this amount discounted as an investment that can be used to justify $v$'s inclusion. To implement this, we keep track of a price, $p(u)$, for every node that represents how much of $c(u)$ is leftover to be paid for. Then, we only need a fall to pay for the price of a path in order to include it in a service, rather than its full cost. 

Overall, the fall will add paths to requests in increasing order of deadline. Naturally, this ordering ensures we satisfy many urgent requests. Though, some requests that aren't urgent may be cheap to add and so should be considered as well. However, the fact that every node included triggers its own fall ensures that requests that aren't globally urgent but are inexpensive to include are added to the service. 

\paragraph{Falls.} Formally, given a tentative service tree $S$ at time $t$, we define a \textit{$v$-fall} to be an iterative procedure that considers requests in $\T_v^t$ for addition to the service. We define $F_v^t$ to be the set of nodes, $\mathcal{A}_v$, constructed by the $v$-fall at time $t$. We will write $P(v \to \rho)$ to mean $P(v \to \rho | S')$, where $S'$ is the tentative service tree right before this path was added to the fall $F_v^t$. If the nodes on a path to a request that are not already included in the service, $S$, can be included without the total price exceeding $v$'s budget, then the nodes are added to the service and the budget reduced by the corresponding price of these nodes. This procedure continues until there are no pending requests in $\T_v^t$ or until the addition of a path to a request exceeds the remaining budget. If a path $P(v \to \rho | S)$ exceeds the remaining budget, then we reduce the price of nodes on this path by the remaining budget. In particular, we evenly divide the remaining budget amongst the price of the path and then reduce each node's price proportional to its contribution to the total price of the path. We refer to these nodes as \textit{overflow nodes} for $F_v^t$. Whenever a fall causes some node's price to change, we think of this fall as investing in that node. This entire process is described as \cref{alg: vfall}. Note we refer to the entire process (\cref{alg: init}, \cref{alg: waterfall}, \cref{alg: vfall}) as well as just \cref{alg: waterfall} as \textsc{Waterfall}.

\paragraph{Waterfall Interpretation.}  \cref{alg: waterfall} behaves like a flow in a flow network defined by $\T$. The sources are the root and other nodes on the path to the due request and the sinks are generally the leaves. Roughly, each $v$-fall is a way of selecting which children of $v$ in the network the flow will continue through. Recursively, all such nodes also send the amount of flow it received down its subtree, and we continue the process until we reach the leaves which cannot send their flow elsewhere. Nodes may have excesses if they don't add enough in a fall, but excess is invested into the overflow nodes for their eventual inclusion. So, if we take future inclusions into account then conservation is preserved. With this flow interpretation, the algorithm behaves like a waterfall where water runs over ledges which act as containers. Whenever enough water fills up on a ledge, it spills over causing a smaller waterfall.

\

\begin{algorithm}[H]
\caption{\textsc{Initialization}}\label{alg: init}
\SetAlgoLined
\For{$v \in \T$}{
    $p(v) \gets c(v)$
}
\end{algorithm}

\

\begin{algorithm}[H]
\caption{\textsc{Waterfall}}\label{alg: waterfall}
\SetAlgoLined
\KwData{Request $\rho$ just became due}
\KwResult{A subtree is transmitted that serves $\rho$ }
 Initiate a service $S = P_{\rho}$ and set $p(v) = c(v)$ for each $v \in P_{\rho}$.\\
 Let $Q$ be a queue for nodes. Enqueue the nodes of $P_{\rho}$ into $Q$\\
 \While{$Q \not = \varnothing$}{
  Dequeue $v$ from $Q$\\
  $\mathcal{A}_v \gets F_v(S)$\\
  For every $u \in \mathcal{A}_v$ enqueue $u$ into $Q$\\
 }
 \textbf{transmit} $S$
\end{algorithm}

\

\begin{algorithm}[H]
\caption{\textsc{$v$-fall}, $F_v(S)$}\label{alg: vfall}
\SetAlgoLined
\SetKwInOut{Input}{Input}\SetKwInOut{Output}{output}
\Input{A tentative service subtree $S$}
\KwResult{The service subtree $S$ is possibly extended to satisfy some requests in $\T_v$ and the set of added nodes is returned. }
 $\mathcal{A}_v = \emptyset, b(v) = c(v)$ \\
 \For{$\rho \in \T_v$ in increasing deadline order}{
  \eIf{$p(v \to \rho | S) > b(v)$}{
    \For{$u \in P(v \to \rho | S)$}{
    $p(u) \gets p(u)(1 - \frac{b(v)}{p(v \to \rho | S)})$
    }
   break;
   }{
   $\mathcal{A}_v \gets \mathcal{A}_v \cup P(v \to \rho | S)$ \\
   $b(v) \gets b(v) - p(v \to \rho | S)$ \\
   \For{$u \in P(v \to \rho | S)$}{
    $p(u) \gets c(u)$
    }
   $S \gets S \cup \mathcal{A}_v$\\
  }
 }
 \Return{$\mathcal{A}_v$}
\end{algorithm}

\begin{theorem}\label{theorem: main}
\textsc{Waterfall} is $D$-competitive on trees of depth $D$.
\end{theorem}

The main approach to analyzing \textsc{Waterfall} will be a charging argument. %For any service of ALG, we need to pay for the nodes it includes by charging them to nodes in OPT's schedule. If a node $v$ in a service of ALG has been included in a previous service of OPT then we will charge its cost to that service of OPT. However, if that node has already been charged in all previous services of OPT, we want to avoid charging it again. In this case, we know a previous service of ALG must have also included $v$. By definition, this service would have constructed a $v$-fall and the higher level nodes present in this $v$-fall can be used to help pay for $v$ in this service. In this way, we can pay for any late nodes in a service of ALG. On the other hand, the service may have satisfied requests that OPT has not yet had a chance to satisfy. For these nodes, we charge their costs to past inclusions or to the costs of their ancestors. Since the root is always late, there always exists a late ancestor that can pay for early nodes. Overall, we will have that every service of ALG is paid for.
Whenever a node is included in a service of ALG, its cost must be paid for by charging to a service of OPT. We think of the amount invested in a node at some time as the total cost it is responsible for. If the node is critically overdue, then it can pay for its own cost and the amount invested in it. Otherwise, a fall investing a node's budget to its descendants can be interpreted as the node delegating the payment of its cost to its descendants. Then, these descendants will either be able to pay for their own costs as well as the ancestors', or will further transfer responsibility of payment to their descendants. Continuing this process, nodes that can pay off the costs will eventually be reached. In particular, the leaves of the tree will eventually be reached and these will be critically overdue. 

However, nodes in the subtree of a critically overdue node may be early and so cannot pay for themselves. The cost of such nodes will also be charged to their critically overdue ancestor. Thus, a critically overdue node will have to pay both for its ancestors (including itself) and its early descendants. This charging idea can be visualized as the evolution of debt in a family tree. The family may continually pass down debt from generation to generation. At some point, the descendants will have accumulated enough wealth to pay-off the debt of their ancestors. Though these descendants might also have children who will accumulate further debt. These unlucky descendants then would be be tasked with paying off debts from both sides, ancestors and descendants.

\subsection{ALG's Structure}

\paragraph{Phases} It will be helpful to break services of ALG into different phases. Suppose that $(O_1,t_1), \ldots, (O_j,t_j)$ are all the services of OPT that include the node $v$ in chronological order of transmission time. Then, for some $i \in [j-1]$ suppose $(A_1,t_1'),\ldots, (A_k,t_k')$ are all the services of ALG that include $v$ and satisfy $t_i \leq t_h' < t_{i+1}$ where $h \in [k]$. We call these services of ALG a $v$-phase and say that $(O_i,t_i)$ starts the $i$-th $v$-phase. For $i = j$, a $v$-phase consists of all the services of ALG containing $v$ that are transmitted at time $t_j$ or later. Note that if a request $\rho$ becomes due in the $i$-th $v$-phase, then $\rho$ must be satisfied by one of $(O_h,t_h)$ for $h \in [i]$. Consequently, if a request, $\rho$, arrives in $\T_v$ during a $v$-phase, it cannot become due before the next $v$-phase since no service of OPT includes $v$ during a $v$-phase and so cannot satisfy $\rho$. Thus, if $v$ is late at time $t$ during a phase, it must have been satisfied earlier. Also, for any service $(A,t)$ in the $i$th $v$-phase, we have that $\pos{v}{t} = t_i$ and $\nos{v}{t} = t_{i+1}$ or $\nos{v}{t} = \infty$ if $i = j$.

 Suppose $v$ is late at the time $t_i$ when $A_i$ is transmitted in the phase. Then, we have $v$ is late at all previous services in the phase. This is because ALG always considers the earliest due requests in $\T_v$ for satisfaction first. Specifically, since late requests cannot arrive in the middle of a phase, we know the request causing ALG to enter $\T_v$, $\rho_i$, must have been available in the previous services of ALG. Hence, the earliest due request in $\T_v^{t_h}$ for $h < i$, $\rho_h$, had an earlier deadline so is also overdue. Formally, $$d^{t_h}(v) = d_{\rho_h} < d_{\rho_i} < \nos{v}{t_i} = \nos{v}{t_h}$$ 
Thus, we can partition any phase into two pieces: an initial segment where $v$ is late and a later segment where $v$ is not late. We call the former services late with respect to the phase, and the latter early. Note it could be that a phase has only early or only late services. However, $r$-phases have only late services since transmissions of the algorithm only happen when a request becomes due. 

Let $v$ be a node included in a service of ALG at time $t$. Recall, we say $v$ is critically overdue at time $t$ if $v$ is late at time $t$ and if for any service of the algorithm's schedule including $v$ at a time $t' \in (t,\nos{v}{t})$ we have that $v$ is not late at time $t'$. In terms of phases, $v$ is critically overdue at time $t$ if the service of ALG at time $t$ is the last late service in its respective $v$-phase. We ``charge'' to critically overdue nodes, $(v,t)$, by charging to the occurrence of $v$ that appears in $\pos{v}{t}$. We have that charging to critically overdue nodes is injective since only one service in any phase may be critically overdue. 

Next, we discuss structural properties of late and critically overdue nodes. First, we note that if $(v,t)$ is late and the $v$-fall satisfies all requests in $\T_v^t$, then $(v,t)$ must be critically overdue. This is clear since a new request must arrive after time $t$ in order for $v$ to become late again and a new service of OPT must be transmitted to satisfy this new request.

\begin{lemma}\label{lemma: latepaths}
Suppose that $v$ is late at time $t$ and that $\rho$ is the earliest due request in $\T_v^t$. Then, for all $u \in P(v \to \rho)$, $u$ is late at time $t$. Furthermore, if $v$ is late at time $t$ and $u \in F_v^t$ is early at time $t$, then $v$ is critically overdue at time $t$.
\end{lemma}

\begin{proof}

If $v$ is late at time $t$ then, by definition, we know that $d^t(v) < \nos{v}{t}$. Since $\rho$ is the earliest due request in $\T_v^t$, we also know that $d^t(v) = d_{\rho}$. If for $u \in P(v \to \rho)$ a request $\gamma \in \T_u^t$ was due earlier than $\rho$, $\gamma$ would also be an earlier due request in $\T_v^t \supseteq \T_u^t$. Hence, we have $\rho$ is the earliest due request in $\T_u^t$, so $d^t(v) = d^t(u)$. Now, by the subtree property of services, if OPT includes $u$ in a service it must also include $v$ in that service implying that $\nos{v}{t} \leq \nos{u}{t}$. Thus, we see that $d^t(u) = d^t(v) < \nos{v}{t} \leq \nos{u}{t}$ and so $u$ is late at time $t$. 

Now, further suppose that $u \in F_v^t$ is early at time $t$. Also, suppose that ALG next includes $v$ at time $t'$. If $v$ is early at time $t'$, then either this service at time $t'$ is in a new $v$-phase or this service is an early service in the same $v$-phase. In either case, the partition of phases into earlier and late parts implies that the service at time $t$ is the last late service for the $v$-phase. Hence, $v$ is critically overdue at time $t$. 

Next, suppose that $v$ is late at time $t'$. Again, if the service at time $t'$ is in a new $v$-phase, $v$ is critically overdue at time $t$. If the service at time $t'$ is in the same $v$-phase as the service at time $t$, we show that $u$ at time $t$ was actually late. In particular, the services at times $t$ and $t'$ being in the same $v$-phase implies that $\pos{v}{t} = \pos{v}{t'}$ and $\nos{v}{t} = \nos{v}{t'}$. We claim that the earliest due request, $\gamma$, in $\T_v^{t'}$ was issued by time $t$. If not, we have that $\pos{v}{t} \leq t < a_{\gamma} \leq d_{\gamma} < \nos{v}{t}$. This means that $\gamma$ arrived and became due between two services of OPT and so OPT never satisfied this request. This is impossible by feasibility of OPT and so $\gamma \in \T_v^t$. Since $u$ was included in the $v$-fall at time $t$, but $\gamma$ was not satisfied by this $v$-fall, we know that $d^t(u) < d_{\gamma}$. But, since $d_{\gamma} = d^{t'}(v)$ and $v$ is late at time $t'$, we have
$$d^t(u) < d_{\gamma} = d^{t'}(v) < \nos{v}{t'} = \nos{v}{t} $$
Finally, applying the subtree property at time $t$, we have that $d^t(u) < \nos{v}{t} \leq \nos{u}{t}$ and so $u$ is late at time $t$. This contradicts that $u$ is early and hence $v$ must be critically overdue at time $t$.

\end{proof}

\paragraph{Direct Investments.} Suppose that $(S,t)$ is a service of ALG with $v \in S$. We say that the pair $(v,t)$ \textit{directly invests} in the pair $(w,t')$ if $w \in F_v^t$ and $t = t'$ or if $w$ is an overflow node for $F_v^t$ and $t'$ is the next time after $t$ that $w$ is included in a service of ALG. The \textit{investment} made, $I_{v}^t(w,t')$, is $p(w)$, $w$'s price during the construction of $F_v^t$, in the former case. In the latter case, the investment made is $\frac{rb(v,t)}{p(v \to \rho)}p(w)$ where $rb(v,t)$ is the remaining budget, $b(v)$, at the end of the $v$-fall's construction and $\rho$ is this earliest request not satisfied by this fall. If $rb(v,t) = 0$, then we do not consider any $0$ investments and say there are no overflow nodes for the fall. We imagine a fall including a node as reducing its price to $0$ before resetting its price to its cost. Then, we see the direct investment in any node is exactly the amount by which the fall reduces the node's price. 

\begin{lemma}\label{lemma: direct}
The total direct investment made by any pair $(v,t)$ is at most $c(v)$ and is strictly less than $c(v)$ only if $F_v^t$ satisfied all requests in $\T_v^t$. In addition, the total direct investments made to a pair $(v,t)$ is at most $c(v)$ and is strictly less than $c(v)$ exactly when $v$ is on the path to the request that triggers the service of ALG at time $t$.
\end{lemma}

\begin{proof}

The total direct investment made by $(v,t)$ is exactly the total price of all nodes included in $F_v^t$ plus the investments made to the overflow nodes. Note that the budget starts at $c(v)$ and is reduced by the price of each node added to the fall in each iteration. Hence, the total investment made by by $(v,t)$ to nodes included in $F_v^t$ is exactly $c(v) - rb(v,t)$. If all requests in $\T_v^t$ are satisfied by $F_v^t$, then this is all the direct investments made and this is clearly at most $c(v)$. If there are overflow nodes, each overflow node, $(u,t')$, is invested exactly $\frac{rb(v,t)}{p(v \to \rho)} p(u)$. Thus, the total investment made by $(v,t)$ to overflow nodes is exactly 
$$\sum_{u \in P(v \to \rho)} \frac{rb(v,t)}{p(v \to \rho)} p(u) = \frac{rb(v,t)}{p(v \to \rho)}\sum_{u \in P(v \to \rho)}p(u) = rb(v,t) $$
So, the total direct investments made by $(v,t)$ is $c(v) - rb(v,t) + rb(v,t) = c(v)$. Hence, as long as there are overflow nodes, the total direct investment will be exactly $c(v)$. So, the only way the total direct investment can be strictly less than $c(v)$ is if there are no overflow nodes, which can only happen if $F_v^t$ satisfied all requests in $\T_v^t$ by definition.

For the second claim, first note that every direct investment made to $(v,t)$ will be when $v$ is a overflow node other than possibly the very last direct investment as $v$ will be included in the service at that point and then cannot be invested in further by definition. Now, the direct investments made to $v$ are exactly the amounts the price of the node is reduced by. Since the price of the node begins at its cost and is reduced until the price becomes $0$ or until the node is included in a service due to a request becoming due, we see the total direct investments made to $v$ is at most $c(v)$. If $v$ is added by a fall then the last investment was its final price and so the total direct investment is simply the final price $p(v)$ plus the amounts the price was reduced by which is $c(v) - p(v)$ for $c(v)$ total direct investments. Hence, we see the total direct investment in $(v,t)$ is less than $c(v)$ exactly when it is not added by some fall. This happens if and only if $v$ was on a path to a request due at time $t$ and so the initial part of \cref{alg: waterfall} added $v$ rather than a fall. 

\end{proof}

\paragraph{Investments.} We say $(v,t)$ \textit{invests} in $(w,t')$ if either $(v,t)$ directly invests in $(w,t')$ or if $(v,t)$ invests in $(u,t'')$ and $(u,t'')$ directly invests in $(w,t')$. Equivalently, the latter case can be defined so that $(v,t)$ directly invests in $(u,t'')$ and $(u,t'')$ invests in $(w,t')$ and we will use both versions interchangeably. The equivalence follows similarly to how we can define a walk in a graph inductively by either looking at the last step or the first step in the walk. In the inductive case, we define $I_v^t(w,t') = \frac{I_v^t(u,t'')}{c(u)} \cdot I_u^{t''}(w,t')$ where one of the investments will be direct and the other will be recursively constructed depending on which definition we use. Intuitively, this is just the investment made to $(w,t')$ by $(u,t'')$ scaled by the relative reduction in $c(u)$ caused by $(v,t)$'s investment to $(u,t'')$.  

We define the \textit{total investment} in $(v,t)$ to be $I(v,t) = \sum_{\substack{(w,t') \text{ invests} \\ \text{in } (v,t)}} I_w^{t'}(v,t)$. Note, that we can calculate the total investment by focusing on direct investments. In particular, $I(v,t) = \sum_{\substack{(w,t') \text{ directly } \\ \text{invests in } (v,t)}} \frac{I(w,t')}{c(w)} I_w^{t'}(v,t)$ which follows since any pair that invests in $(v,t)$ must invest in some pair directly investing in $(v,t)$.
To streamline our arguments we will consider $(v,t)$ as investing in $(v,t)$ and the investment made is $c(v)$.
Note that only ancestors of $v$ can invest in a pair $(v,t)$.

\begin{lemma}\label{lemma: investedin}
For any pair $(v,t)$, we have $I(v,t) \leq L_v c(v)$. 
\end{lemma}

\begin{proof}

We proceed by induction on $L_v$. For the base case, we have $L_v = 1$, so $v = r$. The only ancestor of $r$ is $r$ itself, so the only investment made is from $r$ itself. Hence, the total investment to $(r,t)$ is $I(r,t) = c(r) = L_r c(r)$.

Suppose that $L_v > 1$, so that $v$ is a proper descendant of $r$. If no ancestors of $v$ directly invested in $v$, then again $I(v,t) = c(v) \leq L_v c(v)$. Suppose that $(u_1,t_1), \ldots, (u_k, t_k)$ directly invested in $(v,t)$. Note that each $u_i$ must be a proper ancestor of $v$ and so $L_{u_i} \leq L_v - 1$. The induction hypothesis implies that $I(u_i,t_i) \leq L_{u_i} c(u_i)$ for each $i$. Excluding the investment from $(v,t)$ to itself, the total investment in $(v,t)$ is

\begin{align*}
    I(v,t) - c(v) &= \sum_{i = 1}^k \frac{I(u_i,t_i)}{c(u_i)} I_{u_i}^{t_i}(v,t) \\
    &\leq \sum_{i = 1}^k \frac{L_{u_i} c(u_i)}{c(u_i)} I_{u_i}^{t_i}(v,t) \\ 
    &= \sum_{i = 1}^k L_{u_i} I_{u_i}^{t_i} (v,t)  \\
    &\leq (L_v - 1)\sum_{i = 1}^k I_{u_i}^{t_i}(v,t)
\end{align*}

Lastly, we use the fact from \cref{lemma: direct} that the total direct investment into $(v,t)$ is at most $c(v)$, so this final sum is at most $c(v)$. Thus, we have that $I(v,t) - c(v) \leq (L_v - 1) c(v)$ implying that $I(v,t) \leq L_v c(v)$. 

\end{proof}

We will not only want to look at the amount invested in a pair, but also how much a pair invests in other pairs. To this end, we define $IM(v,t) = \sum_{\substack{(v,t) \text{ invests} \\ \text{in } (w,t'), w \not = v}} I_v^t(w,t')$ to be the total investment made by $(v,t)$ to its proper descendants. Again, we can focus on the pairs that are directly invested in as follows: $IM(v,t) = \sum_{\substack{(v,t) \text{ directly } \\ \text{invests in } (w,t')}} I_v^t(w,t') + \frac{I_v^{t}(w,t')}{c(w)} IM(w,t')$. 

\begin{lemma}\label{lemma: investmentmade}
For any pair $(v,t)$, $IM(v,t) \leq (D-L_v)c(v)$.
\end{lemma}

\begin{proof}

We proceed by induction on $D-L_v$. For the base case, $D-L_v = 0$, so $v$ is a leaf. In this case, $v$ has no proper descendants and so $IM(v,t) = 0 = (D-L_v) c(v)$.

Now, suppose that $D-L_v > 0$ so that $v$ is not a leaf. If $(v,t)$ does not directly invest in any pair we are done, so suppose that $(v,t)$ directly invests in $(w_1,t_1), \ldots, (w_k,t_k)$. By \cref{lemma: direct}, we know the total direct investment made by $(v,t)$ is at most $c(v)$. In symbols, $\sum_{i = 1}^k I_v^t(w_i,t_i) \leq c(v)$. Also, each $w_i$ is a proper descendant of $v$ and so the induction hypothesis implies that $IM(w_i,t_i) \leq (D - L_{w_i}) c(w_i)$ for all $i$. Using these two facts gives,

\begin{align*}
    IM(v,t) &= \sum_{i = 1}^k I_v^t(w_i,t_i) + \frac{I_v^t(w_i,t_i)}{c(w_i)}IM(w_i,t_i) \\
    &\leq  \sum_{i = 1}^k I_v^t(w_i,t_i) + \frac{I_v^t(w_i,t_i)}{c(w_i)}(D - L_{w_i})c(w_i) \\
    &= \sum_{i = 1}^k (D - L_{w_i} + 1) I_v^t(w_i,t_i) \\
    &\leq \sum_{i = 1}^k (D - L_v) I_v^t(w_i,t_i) \\
    &\leq (D- L_v) c(v)
\end{align*}

Above we have used the fact that $L_{w_i} \geq L_v + 1$ since descendants have higher level and so $D - L_{w_i} + 1 \leq D - L_v$.

\end{proof}

\subsection{Charging Scheme}

Now, we begin the construction of the charges we will make. For a set of pairs, $S$, we let $I_v^t(S) = \sum_{(w,t') \in S} I_v^t(w,t')$ be the total investment $(v,t)$ made to the pairs in $S$.

\begin{lemma}\label{lemma: construction}
Let $(S,t)$ be a service of ALG and suppose $v \in S$ is late at time $t$. Then, there exists a set $U_v^t \subseteq \T_v$ of critically overdue nodes such that $I_v^t(U_v^t) = c(v)$.
\end{lemma}

\begin{proof}
We proceed by induction on $D-L_v$.

For the base case, $D-L_v = 0$, so $v$ is a leaf. In this case, we claim $v$ is critically overdue at time $t$ and hence $U_v^t = \{(v,t)\}$ satisfies the requirements as $I_v^t(v,t) = c(v)$. Suppose that there was another service of ALG including $v$ at a time $t' \in (t, \nos{v}{t})$. Further, suppose that $v$ is late at time $t'$. Now, since all requests at $v$ are satisfied after time $t$ by $v$'s inclusion, it must be that some other request $\rho$ arrived at $v$ after time $t$. However, $v$ late at time $t'$ implies $t < a_{\rho} \leq d_{\rho} < \nos{v}{t}$. This implies $\rho$ becomes due before OPT could have satisfied it, contradicting feasibility of OPT. Thus, $(v,t)$ is critically overdue.

Now, suppose that $D-L_v > 0$, so $v$ is not a leaf. If $(v,t)$ is critically overdue, we again have $U_v^t = \{(v,t)\}$ satisfies the requirements. Otherwise, suppose that $v$ is not critically overdue at time $t$. Then, by definition, there must be a time $t' > t$ at which ALG transmits a service, and this service contains $v$, which is late at time $t'$. By the partitioning of phases into early and late parts, we can assume that time $t'$ is the very next time that ALG includes $v$ after time $t$ since at time $t'$, $v$ must be late if it is at a even later time in the phase. 

Consider the request $\gamma$ whose satisfaction caused $v$ to be included at time $t'$ and $\rho$ the last request considered by $F_v^t$. We claim that $\rho = \gamma$. If $\gamma \not = \rho$, then it must be the case that $d_{\gamma} < d_{\rho}$. However, the definition of $v$-falls imply they consider requests in earliest deadline order. Thus, for $\rho$ to be considered before $\gamma$ it must be the case that $\gamma$ arrived after time $t$. But, this implies $t'$ is in another phase and so $v$ was critically overdue at time $t$, a contradiction. Hence, $\rho = \gamma$.

Then, \cref{lemma: latepaths}  implies all the nodes in $P(v \to \rho)$ are late at time $t'$ as $\rho$ is the earliest request in $\T_v^{t'}$. Also, the contrapositive of the last statement in the lemma implies that any $u \in F_v^t$ is late at time $t$ since $(v,t)$ is not critically overdue. Thus, all of the nodes in $F_v^t$ and the overflow nodes $P(v \to \rho)$ are higher level nodes than $v$ that are late so the induction hypothesis applies to these nodes. In particular, if $(w_1,t_1), \ldots, (w_k,t_k)$ are the nodes that $(v,t)$ directly invested in, then we have that for each $i$ there exists a set of critically overdue nodes, $U_{w_i}^{t_i}$, satisfying $I_{w_i}^{t_i}(U_{w_i}^{t_i}) = c(w_i)$. Since $(v,t)$ is not critically overdue, it must be the case that $F_v^t$ did not satisfy every request in $\T_v^t$ as mentioned previously. Thus, by \cref{lemma: direct} we have that the total direct investment made by $(v,t)$ is exactly $c(v)$.  Hence, we have that $\sum_{i = 1}^k I_v^t(w_i,t_i) = c(v)$. Defining $U_v^t = \cup_{i = 1}^k U_{w_i}^{t_i}$, we then have,
\begin{align*}
    I_v^t(U_v^t) &= \sum_{i = 1}^k I_v^t(U_{w_i}^{t_i}) \\
    &= \sum_{i = 1}^k \frac{I_v^t(w_i,t_i)}{c(w_i)}I_{w_i}^{t_i}(U_{w_i}^{t_i}) \\
    &= \sum_{i = 1}^k I_v^t(w_i,t_i) = c(v)
\end{align*}
Also, $U_v^t \subseteq \T_v$ contains only nodes that are critically overdue.

\end{proof}

Fix any service $(S, t)$ of ALG. \cref{lemma: construction} implies for any late node $v \in S$, we can pay for the $c(v)$ in this service by charging the critically overdue nodes in $U_v^t$ the amount $(v,t)$ invested in them. This will be the charges we make for late nodes. Note, at this point the total charge to any critically overdue node, $(v,t)$, will be at most $L_v$. This follows since we at most charge a node the total investment made in it and this is at most $L_v c(v)$ by \cref{lemma: investedin}.

We also need to pay for early nodes of a service, which we do next. The analysis above essentially covers the cost of all nodes ``above" a node $v$, i.e. its ancestors. Now, we want to bound the cost of all nodes ``below" a node $v$, i.e. its descendants. To continue the charging strategy, we will charge every critically overdue node, $(v,t)$, for all of the early nodes that it invests in. This is at most the totality of all investments made by $(v,t)$ so this adds a charge of at most $D-L_v$ to $(v,t)$ according to \cref{lemma: investmentmade}. Thus, each node is charged $L_v$ for nodes above it and $D-L_v$ for nodes below it for a total of $D$. We just need to argue every cost is in fact covered by this strategy.

\begin{lemma}\label{lemma: coverage}
Let $(S,t)$ be a service of ALG. Then, the charging strategy pays for all the costs of nodes in $S$.
\end{lemma}

\begin{proof}

Any late node is clearly paid for by the definition of the charging scheme and \cref{lemma: construction}. Next, we argue all early nodes are also paid for. Note any node on a path to a due request cannot be early by \cref{lemma: latepaths}. Thus, if a node is early it cannot be on the path to a request that triggers a service of ALG. Consequently, \cref{lemma: direct} implies the total investments to an early node must exactly equal its cost.

We proceed by induction on $L_v$ to show if $v$ is early at time $t$ then the cost of $v$ is exactly paid for by charges to some critically overdue nodes as dictated by the charging scheme. For the base case, $L_v = 1$ and we have that $v = r$. We know $r$ is never early by definition of ALG only including $r$ when a request becomes due and so this case vacuously holds. 

Suppose that $L_v > 1$. Let $(u,t')$ be a pair that directly invested in $(v,t)$. If $(u,t')$ is critically overdue, then we know by the charging scheme that $(u,t')$ pays for the investment $I_u^{t'}(v,t)$ that it made to $(v,t)$. We just need to ensure that investments of the other pairs that are not critically overdue were paid for. First, we claim all non-critically overdue pairs that invested in $(v,t)$ are early. For the sake of contradiction, suppose that $(u,t')$ directly invested in $(v,t)$ and $(w,t')$ is late but not critically overdue. 
\begin{itemize}
    \item If $v \in F_u^{t'}$ is early at time $t'$, then \cref{lemma: latepaths} implies that $(u,t')$ is critically overdue, a contradiction.
    \item If $F_u^{t'}$ does not add $v$ at time $t$, then $v$ was an overflow node for $(u,t')$. Since $(u,t')$ is late but not critically overdue, we know the next time $t''$ that ALG includes $u$ after time $t'$ is in the same $u$-phase as time $t'$. However, since $v$ was an overflow node at time $t'$, we know that the earliest request in $\T_u^{t'}$ that is not satisfied is in $\T_v^{t'}$. 
    
    \begin{itemize}
        \item If the next time we enter $u$ the path includes $v$, then this happens at time $t''$ and so $t = t''$. Then, \cref{lemma: latepaths} implies $(v,t)$ is late since its on the path to the earliest due request in $\T_u^{t}$. 
        \item If the next time we enter $u$ the path does not include $v$, it must be that a new request arrived after time $t'$ causing $u$ to be late. However, late requests cannot arrive in the middle of a phase as it would become due before OPT includes $u$ again. 
    \end{itemize}
\end{itemize}
%If $v \in F_u^{t'}$ is early at time $t'$, then \cref{lemma: latepaths} implies that $(u,t')$ is critically overdue, a contradiction. If $F_u^{t'}$ does not add $v$ at time $t$, then $v$ was an overflow node for $(u,t')$. Since $(u,t')$ is late but not critically overdue, we know the next time $t''$ that ALG includes $u$ after time $t'$ is in the same $u$-phase as time $t'$. However, since $v$ was a overflow node at time $t'$, we know that the earliest request in $\T_u^{t'}$ that is not satisfied is in $\T_v^{t'}$. If the next time we enter $u$ the path includes $v$, then this happens at time $t''$ and so $t = t''$. Then, \cref{lemma: latepaths} implies $(v,t)$ is late since its on the path to the earliest due request in $\T_u^{t}$. Otherwise, if $v$ is not included on this earliest path, it must be that a new request arrived after time $t'$ causing $u$ to be late. However, late requests cannot arrive in the middle of a phase as it would become due before OPT includes $u$ again. 
In all cases, we reach a contradiction and so if $(u,t')$ is not critically overdue then it must be early.

Now, for each $(u,t')$ that is early we have that $u$ is an ancestor of $v$ so the induction hypothesis implies that some charges to critically overdue nodes paid for $c(w)$. In particular, by our charging scheme this means there were critically overdue nodes $(w_1, t_1), \ldots , (w_j, t_k)$ satisfying $\sum_{i = 1}^k I_{w_i}^{t_i}(u,t') = c(u)$. However, since $v$ is an early descendant of $u$, $v$ is also an early descendant of each $w_i$. Hence, our charging scheme would also ensure that $(w_i, t_i)$ paid for its investment to $(v,t)$ since it paid its investment to $(u,t')$. By definition of the investments, the total investment made to $(v,t)$ by these critically overdue nodes is exactly
$$\sum_{i = 1}^k I_{w_i}^{t_i}(v,t) =  \sum_{i = 1}^k \frac{I_{w_i}^{t_i}(u,t')}{c(u)} I_u^{t'}(v,t) = \frac{I_u^{t'}(v,t)}{c(u)} \sum_{i = 1}^k I_{w_i}^{t_i}(u,t') = I_u^{t'}(v,t)$$
Now, since this holds true for all early pairs investing in $(v,t)$, we have that all direct investments to $(v,t)$ are paid for by our charging scheme. But the total direct investment for an early node is its cost, and so $c(v)$ is paid for by the charges. Thus, the charging scheme pays for all of the costs of the service $(S,t)$.

\end{proof}

Summarizing the charging strategy above gives a proof of \cref{theorem: main}.

\begin{proof}

For any service $(S,t)$ of ALG, we apply \cref{lemma: construction} to create sets of critically overdue nodes to charge. In particular, each node $v \in S$ that is late at time $t$ we pay for by charging each pair in $U_v^t$ the investment $(v,t)$ made to each. This pays for every late node as $I_v^t(U_v^t) = c(v)$. By \cref{lemma: investedin} the total investments made to $(v,t)$ is at most  $L_v c(v)$ and so the charge to any critically overdue node is at most $L_v$ to pay for late nodes through investments. Then, we charge to each critically overdue node the total amount of investments it made to early nodes. These charges add a total of charge of at most $D-L_v$ to any node $v$ as the total investments made by $(v,t)$ is at most $(D-L_v)c(v)$ by \cref{lemma: investmentmade}. Hence, the total charge to any critically overdue node is at most $D$. Lastly, \cref{lemma: coverage} implies every node in a service of ALG is paid for by this charging. Thus, the total cost of ALG's schedule is at most $D$ times the total cost of OPT's schedule. Consequently, ALG is $D$-competitive. 

\end{proof}

One thing to note about the above analysis is that its basically tight. Most critically overdue nodes are charged exactly $D$. This optimizes the charging possible if we restrict to only charging to nodes OPT has included prior to a fixed time. The only way to improve upon this analysis would be to allow charging early nodes as well. However, recursive instances give evidence that no such approach would be possible as the adversary could change the instance if it sees ALG early-aggregate many nodes. 
%However, recursive instances like the one given for \textsc{Double} give evidence that no such analysis would be possible as the adversary could change the instance if it sees ALG early-aggregate many nodes. 

%Similar recursive examples seem to indicate a recursive based algorithm is necessary for MLAPD. In particular, whenever a node is included in the service it is critical to consider adding nodes in its subtree. Otherwise, the adversary could ensure the next due request appears in that subtree and this causes the algorithm to pay much more than the optimal. However, such recursive strategies guarantee an eager-aggregation cost of at least $D$ as the nodes added recursively may have all been too early and we added some in every level. In a sense, recursive instances allow us to conclude that the worst early aggregation factor lower-bounds the competitive ratio for that instance. This gives some indication of why the competitive ratio for MLAPD may depend on $D$ from an adversarial standpoint.

At a higher level, there is a clear trade-off between the eager and lazy aggregation costs and neither cost seems to be more dominant than the other. So, perfectly balancing both likely would give the best competitive ratio for the problem and this is exactly what \textsc{Waterfall} does. Specifically, the maximum early aggregation is $D$ resultant from the early aggregation made by the root and the maximum lazy aggregation will come from the leaves which will also have a charge of $D$. More generally, the sum of the early charges and the late charges add to $D$ as the analysis shows. Although they may not always be exactly $D/2$ each, we always have the early and lazy aggregation costs are in equilibrium with each other balancing out to the fixed number $D$. Depending on the tree structure and the requests received, strictly more eager aggregation or strictly more lazy aggregation may be more beneficial and so this more general type of balancing seems necessary. This further indicates that $D$ may be the best competitive ratio possible for MLAPD. 

%Exactly balancing over aggre and under aggr not possible the sum tradeoff is best since leaves never overaggr and root never underaggrs but this analysis gives under + over = D and at level D/2 this does get the exact tradeoff.

\section{Conclusions}

%We improve beyond past results by avoiding the bottleneck of $L$-decreasing trees and instead attack the problem directly. Several structural properties of the problem are explored that allow us to maximize the effectiveness of the algorithm \textsc{Waterfall}.

One prominent open question left by this paper is whether $D$ is the true competitive ratio for MLAPD. In particular, a critical result would be new lower-bounds for MLAPD that improve upon the $2$ lower-bound for the $D = 2$ case. In addition, determining the difference in complexity between the delay and deadline cases is a critical open problem. 

For $D = 1$, we know the optimal competitive ratio for MLAPD is $1$ and that for MLAP it is $2$. For $D = 2$, we know that the optimal competitive ratio for MLAPD is $2$ and for MLAP it is essentially $3$ (improving the known lower-bound of $2.754$ to $3$ is still unresolved). Unless the problems for larger $D$ differ greatly than that for smaller $D$, this pattern would suggest that the optimal competitive ratio for MLAP is exactly $1$ more than that for MLAPD. 

In fact, \cite{OPD} essentially took an algorithm for JRPD and converted it into an algorithm for JRP using a primal dual framework. The resultant algorithm had competitive ratio $1$ more than the competitive ratio for the algorithm for JRPD. If this technique could be extended for $D > 2$ then using \textsc{Waterfall} as the MLAPD algorithm could result in a $D+1$ competitive algorithm for MLAP, which would improve on the best known $O(D^2)$-competitive algorithm. This is especially promising since the JRPD algorithm used for that result behaves almost exactly as \textsc{Waterfall} does on trees of depth $2$. 

Another possible approach would be to extend \textsc{Waterfall} to work for the delay case directly. By simply replacing deadlines with maturation times and making minor adjustments, the modified \textsc{Waterfall} may in fact be $D+1$-competitive. However, unlike the nice structural properties of phases being partitioned into late and early parts, no such structure is clear for the delay version. Thus, a different analysis strategy must be employed absent new structural insights. 

%For $D = 1,2$, we have tight upper and lower bounds for the competetive ratio (at least for the deadline case) as the known TCP-acknowledgment problem and JRP encapsulate exactly MLAP for these depths. However, even for $D = 3$, no better competitive algorithm is known, and this special case should be easier than the general problem. Also, for these smaller depths, the delay and deadline problem have very close competitive ratios (only differing by 1). On the other hand, for general depths we have a much larger difference in competitive ratio of the two problems, so it seems an improvement to the algorithm in \cite{GF} for delays should be possible, maybe even by using techniques from \cite{D} for the deadline case. Whether $D$ should be the best competitive ratio for the deadline case, or whether a constant independent of $D$ could be achieved is also an interesting direction, though likely much harder. All the lower bounds for MLAP rely on a variation of the problem, and this variation cannot be used to give better lower bounds as shown in \cite{OG}. Thus, finding stronger lower bounds, or lower bounds that don't use the variant used before would be interesting. In addition, looking at special types of trees, especially $L$-decreasing trees, to see which are truly the harder and easier instances for the problem would also be useful. Lastly, it would be interesting to see if randomness helps greatly either for the offline or online version of the problem. 

\printbibliography

\newpage 

\appendix

\section{Proofs for section~\ref{section: paths}}\label{sec: pathappendix}

We begin with the proof of \cref{lemma: doubling}.
\begin{proof}
In symbols, the assumptions state that $d_{\rho_1} < \ldots < d_{\rho_k}$, $c(S_i) < c(O)$, and $P_{\rho_i} \subseteq P_{\rho_k} \subseteq O$ for all $i \leq k$. We will show that each interval transmitted by ALG has cost at at least double the previously transmitted interval, and use this to show that in reverse order each interval must be at least a power of two smaller than the cost of $O$. 

Let $P(\rho_{i-1} \to \rho_i) = P_{\rho_i} \setminus P_{\rho_{i-1}}$ be the nodes that need to be added to $P_{\rho_{i-1}}$ in order to reach the node at which $\rho_i$ was issued. Since $\rho_i$ triggered a service of ALG, we know that $\rho_i$ was not satisfied by previous services of ALG. Then, we know that the nodes in $P(\rho_{i-1} \to \rho_i)$ were not added to the service. By definition of ALG, this implies that $c(\rho_{i-1} \to \rho_i) > c(P_{\rho_{i-1}})$. Thus, $$c(P_{\rho_i}) = c(P_{\rho_{i-1}}) + c(\rho_{i-1} \to \rho_i) > 2c(P_{\rho_{i-1}})$$
In other words, the intervals of the requests are at least doubling each time a request is not satisfied. This implies that $c(P_{\rho_{i-1}}) < \frac{1}{2}c(P_{\rho_i})$ holds for all $i > 1$. As $c(P_{\rho_k}) \leq 2^{-(k-k)}c(P_{\rho_k})$, an easy induction shows that $c(P_{\rho_i}) \leq 2^{-(k-i)}c(P_{\rho_k})$ holds for all $i$. The definition of ALG then implies that $c(S_i) \leq 2c(P_{\rho_i}) \leq 2^{-(k-i)+1}c(P_{\rho_k})$. 

Suppose $\gamma$ is the earliest request due at the node $v$ that is the furthest node from the root contained in $O$. As $\rho_k$ was satisfied by $O$ by assumption, and $t \leq t_k$, it must be the case that the request $\gamma$ arrived before the deadline of $\rho_k$ and so has been issued at the time of $S_k$'s construction. Then, we know that $\gamma$ was not satisfied by $S_k$ otherwise $v$ being the last node in $O$'s interval would imply $c(S_k) \geq c(O)$. Thus, the same arguments from before imply that $c(P_{\rho_k}) \leq \frac{1}{2} c(P_{\gamma}) = c(O)$. Thus, for all $i$ we have that 
$$c(S_i) \leq 2^{-(k-i)+1}c(P_{\rho_k}) < 2^{-(k-i)+1}(\frac{1}{2}c(O)) = 2^{-(k-i)}c(O)$$
The total cost of these intervals ALG transmitted is
\begin{align*}
    \sum_{i = 1}^k c(S_i) = \sum_{i = 1}^k c(S_{k-i+1}) 
    = \sum_{i = 0}^{k-1} c(S_{k-i}) 
    < \sum_{i = 0}^{k-1} 2^{-i} c(O) 
    = (2 - 2^{-k})c(O)
\end{align*}
\end{proof}

Next, we show the improved competitive ratio for \textsc{Double}.

\begin{proof}
We simply improve the upper bound on $k$ from $D$ to $\frac{D}{2}-1$ to achieve the claimed competitiveness. Notice that the worst charge to a service of OPT happens when both a larger service and smaller services of ALG both charge to it. In this case, to achieve a $2$ charge from the larger service it must be that at least two nodes are not included by the last small service (one that is due and another that is added too early). So, we know that $k \leq D - 2$. Now, in one situation we could have every request is at its own node and each node has very large cost and so $k = D-2$. However, this will not achieve the largest charge. In particular, this means $c(S_i) = c(\rho_i)$ for each $i$ and so instead of $c(S_{k-i}) \leq 2^{-i} c(O)$ we have $c(S_{k-i}) \leq 2^{-(i+1)} c(O)$, which overall yields a decrease of at least $1$ in the charge. 

In order for $c(S_k)$ to be very close to the $c(O)$, which happens in the worst case, we need at least an additional node between the node $\rho_k$ is issued at and the last node served by $O$. This pushes the total cost of $S_k$ from $\frac{1}{2} c(O)$ closer to $c(O)$ in exchange for one fewer interval than $D-2$ could be transmitted in this range. However, the overall increase in cost is positive. In particular, even if we added infinitely many intervals, the total cost would be $\sum_{j = 1}^\infty 2^{-j} = 1$, whereas just losing one later interval also pushes the cost to $1$, so there is no improvement for choosing more intervals over a longer initial interval. In fact, if the path is finite, the longer interval is a strict improvement. Similarly, for any request $\rho_{k-i}$ where we would want to exchange a longer service $S_{k-i}$ for fewer total intervals, this is an improvement as $\sum_{j = i}^\infty 2^{-j} = 2^{1-i}$ which is exactly the largest we could make $S_{k-i}$ by sacrificing one smaller interval. Hence, the worst case occurs when each smaller interval is as large as possible, which requires at least one node to exist between any two request nodes. Thus, we remove half of the nodes from consideration giving that $k \leq \frac{D}{2} - 1$. Thus, ALG is $4 - 2^{1 - \frac{D}{2}}$-competitive.
\end{proof}

We note that the analysis above actually depends not only on $k$, but the minimal difference, $\epsilon$, between intervals charged to a longer service. The most accurate competitive ratio that could be achieved by the above analysis is likely closer to $4 - \epsilon k 2^{1-k}$ by taking the minimal differences into account appropriately. This indicates an even better competitive ratio may be possible for path instances.
%(here $\epsilon = 1$ since the overflow that would be caused by adding the next node to one of ALG's service's is $1$) though this refinement is not a significant improvement in most graphs and seems to need a much more involved analysis. 

%This kind of recursive counterexample shows the danger of eager-aggregation. If ALG aggregates too much (in this case includes node $D$ unnecessarily), the adversary could exploit that eager-aggregation. In this example, ALG keeps adding node $D$ when it should not. By choosing requests appropriately, the adversary can guarantee that ALG does the same eager-aggregation arbitrarily many times. On the other hand, OPT serves the same node $D$ only a single time. Hence, there is no hope for ALG to charge this huge cost to the single occurrence of $D$ appearing in OPT. Thus, ALG is forced to charge this cost to smaller cost nodes causing a larger charge than necessary. These types of instances show that it is critical to balance the eager-aggregation and lazy-aggregation of the algorithm precisely.

\end{document}